\documentclass[11pt, letterpaper]{elsarticle}
\usepackage{amsmath}
\usepackage{amssymb}
\usepackage{amsfonts}
\usepackage{graphicx}
\usepackage{multirow}
\usepackage{latexsym}
\usepackage{epic}
\usepackage{url}
\usepackage{soul}
\usepackage{afterpage}
\usepackage[unicode, pdftex]{hyperref}
\usepackage{multicol}
\usepackage{bbm}
\usepackage{setspace}
\usepackage{enumitem}
\usepackage{cleveref}
\usepackage[toc]{appendix}
\usepackage[norelsize,ruled, lined,linesnumbered]{algorithm2e}
\usepackage{amsthm}
\usepackage[symbol]{footmisc}
\usepackage[mathscr]{euscript}
\usepackage[left=0.9in,top=0.9in,right=0.9in,bottom=1in,nohead]{geometry}
\usepackage{empheq}
\usepackage{mathtools}

\DeclareMathOperator{\argmin}{argmin}

\makeatletter
\newtheorem*{rep@theorem}{\rep@title}
\newcommand{\newreptheorem}[2]{%
	\newenvironment{rep#1}[1]{%
		\def\rep@title{#2 \ref{##1}}%
		\begin{rep@theorem}}%
		{\end{rep@theorem}}}
\makeatother
\SetAlgoNlRelativeSize{-1}
\SetKwFor{For}{for}{}{end}
\SetKwFor{While}{while}{}{end}
\newreptheorem{example}{Example}
\newtheorem{theorem}{Theorem}

\newtheorem{lemma}{Lemma}

\newtheorem{definition}{Definition}
\newtheorem{corollary}{Corollary}

\usepackage{xcolor}
\definecolor{forestgreen}{RGB}{34,139,34}
\usepackage[colorinlistoftodos,prependcaption,textsize=tiny]{todonotes}
\usepackage{xargs}
\usepackage{booktabs}
\usepackage{caption}
\usepackage{subcaption}

\usepackage{lineno}
\usepackage{float}
\usepackage[section]{placeins}
\usepackage{arydshln}
\usepackage{latexsym}
\usepackage{epic}
\usepackage{amsmath,amsthm,amssymb}
\usepackage{graphicx}
\usepackage{url}
\usepackage{pgfplots}\usetikzlibrary{plotmarks}
\usepackage{soul}
\usepackage{afterpage}
\usepackage{cleveref}
\usepackage{bbm}
\usetikzlibrary{decorations.markings}
\usepackage{setspace}
\usepackage{appendix}
\usepackage{comment}
\DeclareGraphicsExtensions{.pdf,.png,.jpg,.bmp}

\allowdisplaybreaks

\begin{document}

	\begin{frontmatter}
\title{\onehalfspacing On the Tightness of Standard Relaxations for \\ 
	Mixed-Integer Bilevel Linear Programs}
		
		\author[label1]{Sergey S.~Ketkov\footnote[2]{Corresponding author. Email: sergei.ketkov@business.uzh.ch; phone: +41 078 301 85 21.}}
		\author[label1]{Oleg A. Prokopyev}
		\address[label1]{Department of Business Administration, University of Zurich, Zurich, 8032, Switzerland}
		\begin{abstract}
			\doublespacing
	   Exact algorithms for solving \textit{mixed-integer bilevel linear programs} (MIBLPs) typically rely on sequences of lower and upper bounds that converge to the optimal value. These procedures are commonly initialized using the \textit{single-level relaxation} (SLR), obtained by omitting the follower's optimality condition and solving the resulting single-level optimization problem. In this paper, we investigate whether, for broad classes of MIBLPs, the resulting standard bounds admit \emph{uniform} improvements that can be computed within the same computational complexity regime. For \textit{pure continuous} bilevel linear programs, we show that, unless $P = NP$, neither the SLR-based lower bound nor its associated upper bound can be uniformly improved in polynomial time, even for the class of min-max problems. We then extend this analysis to the class of \textit{pure integer} min-max~bilevel linear programs under the assumption that the polynomial hierarchy does not collapse. 
	   First, we show that the continuous relaxation of the SLR admits no uniform polynomial-time computable improvement. We then prove that neither the SLR itself nor its associated upper bound admits a uniform improvement by a polynomial-time algorithm with access to  a mixed-integer linear programming (MILP) oracle.
	   Importantly, this rules out uniform improvements by iterative MILP-based approaches, including cutting-plane-based and decomposition algorithms. Overall, our results demonstrate that the SLR-based bounds are, in a complexity-theoretic sense, unimprovable systematically within their natural computational regimes.
		\end{abstract}
		
	\end{frontmatter}
	\doublespacing
\section{Introduction}	
\textit{Bilevel optimization} addresses hierarchical optimization problems involving an upper-level decision-maker (the \textit{leader}) and a lower-level decision-maker (the \textit{follower}). The leader makes a decision first, optimizing its objective function and anticipating that the follower subsequently solves its own optimization problem, parameterized by the leader's decision. Comprehensive surveys of bilevel optimization problems and their applications can be found in \cite{Colson2007, Kleinert2021, Sinha2017}.

In this paper, we consider a class of optimistic \textit{mixed-integer bilevel linear programs} (MIBLPs) of the~form:
\begin{subequations} \label{MIBLP} 
	\begin{align} [\textbf{BP}]:\quad z^*_{\,BP} := \min_{\mathbf{x},\mathbf{y}^*}\;& \mathbf{a}^\top \mathbf{x} + \mathbf{d}^\top \mathbf{y}^* \label{obj: leader} \\ \text{s.t.}\;& \mathbf{x} \in X \label{cons: leader} \\
	&\mathbf{y}^* \in \argmin_{\,\mathbf{y} \in Y(\mathbf{x})} \mathbf{g}^\top \mathbf{y}, \label{cons: follower} \end{align} \end{subequations}
where
\begin{subequations}
	\begin{align} 
		& X := \Big\{\mathbf{x} \in \mathbb{R}_+^{n_1}\times \mathbb{Z}_+^{n_2}: \mathbf{H} \mathbf{x} \leq \mathbf{h} \Big\}, \label{eq: leader's feasible set} \\ &	Y(\mathbf{x}) := \Big\{\mathbf{y} \in \mathbb{R}_+^{m_1}\times \mathbb{Z}_+^{m_2}: \mathbf{L} \mathbf{x} + \mathbf{F} \mathbf{y} \leq \mathbf{f}\Big\} \label{eq: follower's feasible set}
	\end{align}	
\end{subequations}
are, respectively, the leader's and the follower's feasible sets. The parameters of the problem are given by
$\mathbf{H} \in \mathbb{Q}^{p \times n}$,
$\mathbf{L} \in \mathbb{Q}^{q \times n}$, 
$\mathbf{F} \in \mathbb{Q}^{q \times m}$, $\mathbf{h} \in \mathbb{Q}^{p}$, $\mathbf{f} \in \mathbb{Q}^{q}$, $\mathbf{a} \in \mathbb{Q}^{n}$, $\mathbf{d} \in \mathbb{Q}^{m}$ and~$\mathbf{g} \in \mathbb{Q}^{m}$, where $n := n_1 + n_2$ and $m := m_1 + m_2$; see, e.g., \cite{Audet1997,Kleinert2021}. In particular, [\textbf{BP}] is referred to as a \textit{min-max} problem when~$\mathbf{g} = -\mathbf{d}$.

MIBLPs of the form~[\textbf{BP}] arise naturally in the contexts of network design \cite{Ben1992,Fontaine2014}, electricity markets and energy systems~\cite{Baringo2012, Wogrin2020}, supply chain management~\cite{Yue2017}, as well as interdiction of critical infrastructure~\cite{Borrero2019, Caprara2016}.   
Thus, in line with the majority of the bilevel optimization literature, [\textbf{BP}] is formulated in the \textit{optimistic} sense. In other words, among all follower-optimal solutions in~(\ref{cons: follower}), the one minimizing the leader's objective function~(\ref{obj: leader}) is selected; see, e.g., \cite{Dempe2002,Wiesemann2013}.
Furthermore, we assume that the leader's feasible set $X$ does not depend on the follower's optimal solution, i.e., \textit{coupling constraints} are absent.
As outlined below, these two assumptions allow us to construct standard lower and upper bounds for the optimal objective function~value~$z^*_{\,BP}$.  

%


\subsection{Single-Level Relaxation} \label{sec: intro SLR}
Most existing solution methods for~[\textbf{BP}] rely on decomposition and cutting-plane-based techniques. These methods iteratively construct sequences of valid lower and upper bounds that converge to the optimal objective function value~$z^*_{\,BP}$; see, e.g.,~\cite{Caprara2016,Fischetti2017,Moore1990} and the survey in \cite{Kleinert2021}. In particular, the initial bounds are typically obtained from the \textit{single-level relaxation} of [\textbf{BP}] (also known as \textit{high-point relaxation}), defined by dropping the follower's optimality condition in (\ref{cons: follower}), i.e.,
 \begin{subequations} \label{SLR} 
	\begin{align} [\textbf{SLR}]:\quad z^*_{\,SLR} := \min_{\mathbf{x},\mathbf{y}}\;& \mathbf{a}^\top \mathbf{x} + \mathbf{d}^\top \mathbf{y}  \\ \text{s.t. } & \mathbf{x} \in X, \quad
		\mathbf{y} \in Y(\mathbf{x}).
\end{align} \end{subequations}

Under standard regularity assumptions, $z^*_{\,SLR}$ is a valid lower bound for $z^*_{\,BP}$. Moreover, in the absence of coupling constraints, one may also construct an associated upper bound. To this end, let~$\hat{\mathbf{x}}^*$ denote a leader-optimal solution of [\textbf{SLR}] and let $\hat{\mathbf{y}}^*$ be an associated follower's optimal response, i.e.,
\begin{equation} \label{eq: follower's optimal response}
	\hat{\mathbf{y}}^* \in \argmin_{\,\mathbf{y}} \Big\{ \mathbf{g}^\top \mathbf{y}:\; \mathbf{y} \in Y(\hat{\mathbf{x}}^*) \Big\}.
\end{equation} 
Whenever well-defined, the pair $(\hat{\mathbf{x}}^*, \hat{\mathbf{y}}^*)$ is bilevel feasible for~[\textbf{BP}], and hence 
\begin{equation} \label{eq: upper bound}
	\hat{z}_{\,U}(\hat{\mathbf{x}}^*, \hat{\mathbf{y}}^*) := \mathbf{a}^\top \hat{\mathbf{x}}^* + \mathbf{d}^\top \hat{\mathbf{y}}^*
\end{equation}
is a valid upper bound for $z^*_{\,BP}$. As a result, we have \[z^*_{\,SLR} \leq z^*_{\,BP} \leq \hat{z}_{\,U},\]
where the initial gap $\hat{z}_{\,U} - z^*_{\,SLR}$ often affects the computational effort required for solving~[\textbf{BP}].


 
 Although~[\textbf{SLR}] is known to provide arbitrarily loose bounds for certain problem classes \cite{Caprara2016,Kleinert2021}, the complexity-theoretic properties of the lower bound $z^*_{\,SLR}$ and the upper bound $\hat z_{\,U}$ have, to the best of our knowledge, not been systematically investigated. 
 This naturally raises the question of whether the standard SLR-based bounds admit, in a sense, \textit{strict} and \textit{uniform} improvements with respect to a broad class of bilevel problems.
 
 Related complexity-theoretic limitations on improving standard optimization bounds have been established by Busygin and Pasechnik~\cite{Busygin2006} for the maximum independent set problem. For example, they show that, unless $P=NP$, no polynomial-time computable upper bound on the independence number can be, in a sense, tighter than the Lov{\'a}sz-number bound \cite{Lovasz1979}. Kahruman-Anderoglu et~al.~\cite{Kahruman2016} subsequently introduce the related notion of \textit{provably best} construction heuristics and establish that several standard greedy heuristics for the maximum clique problem admit no strict polynomial-time improvement simultaneously for all relevant graph instances. 
 
 Motivated by this perspective, we formulate the following research question:
 

\begin{itemize} 
	\item[$ $] [\textbf{RQ}]: Given a class of bilevel problems $\mathcal{C}$, do there exist alternative
	lower and upper bounds $\beta_L$ and $\beta_{\,U}$, computable within the same computational complexity
	regime as $z^*_{\,SLR}$ and $\hat z_{\,U}$, respectively, such that, for \textit{every} instance in $\mathcal{C}$, \[
	z^*_{\,SLR} \leq \beta_{\,L} \leq z^*_{\,BP}
	\leq \beta_{\,U} \leq \hat z_{\,U} ,
	\]
   with $\beta_{\,L} > z^*_{\,SLR}$ whenever $z^*_{\,SLR} < z^*_{\,BP} $, and
   $\beta_{\,U} <\hat z_{\,U}$ whenever $z^*_{\,BP} <\hat z_{\,U}$. 
\end{itemize}
Whenever such bounds $\beta_L$ or $\beta_{\,U}$ do not exist, we say that the corresponding bound $z^*_{\,SLR}$ or $\hat z_{\,U}$, respectively, cannot be uniformly improved and is \textit{provably best} with respect to the class $\mathcal{C}$.

As we demonstrate later in Sections \ref{sec: blp} and \ref{sec: miblp}, the research question~[\textbf{RQ}] is particularly relevant for the \textit{pure continuous} case ($n_2 = m_2 = 0$) and the \textit{pure integer} case ($n_1 = m_1 = 0$). In these cases, standard regularity conditions guarantee the existence of an optimal solution to [\textbf{BP}], while there exists a complexity gap between computing the initial bounds, $z^*_{\,SLR}$ and $\hat{z}_{\,U}$, and solving the underlying bilevel problem. Finally, we note that when the leader's continuous variables in [\textbf{BP}] affect the follower's feasible set~(\ref{eq: follower's feasible set}) and $m_2 \neq 0$, an optimal solution to [\textbf{BP}] is not attained in general, even when both the leader's and the follower's feasible sets are nonempty and compact; see, e.g.,~\cite{Koppe2010}. 



\subsection{Our Contributions} \label{sec: intro contributions} We first analyze the pure continuous case ($n_2 = m_2 = 0$), in which [\textbf{BP}] reduces to a \textit{bilevel linear program} (BLP). Notably, optimistic BLPs constitute the simplest canonical class of bilevel problems for which the single-level relaxation~[\textbf{SLR}] and the associated upper bound~(\ref{eq: upper bound}) are naturally defined. In addition, when $m_2 = 0$, replacing the follower's problem in (\ref{cons: follower}) with its Karush--Kuhn--Tucker (KKT) optimality conditions and linearizing the resulting complementarity constraints yields a standard single-level mixed-integer linear programming (MILP) reformulation of [\textbf{BP}]; see, e.g.,~\cite{Audet1997,Fontaine2014,Zare2019}. The~linear programming relaxation of this MILP reformulation, in turn, yields an alternative lower bound for~[\textbf{BP}], which we also examine in the context of the research question [\textbf{RQ}]. 

Importantly, all our complexity-theoretic results are established for the restrictive class of \textit{min-max} problems. In this setting, [\textbf{RQ}] is particularly relevant, as the single-level relaxation [\textbf{SLR}] effectively reverses the follower's optimization direction in (\ref{cons: follower}), thereby providing relatively weak lower bounds; see, e.g.,~\cite{Caprara2016,Kleinert2021}. 

Our contributions for the class of min-max BLPs can be summarized as follows: 
\begin{itemize}
	\item We show that, unless $P=NP$, the standard lower bound
	$z^*_{\,SLR}$ is provably best
	(Theorem~\ref{theorem 1}). A similar result is then established for the lower bound obtained from the linear programming relaxation of the KKT-based MILP reformulation (Corollary~\ref{corollary 1}).
	
	\item We show that, when [\textbf{SLR}] admits multiple leader-optimal
	solutions, computing the tightest upper bound of the form (\ref{eq: upper bound}) is $NP$-hard (Theorem~\ref{theorem 2}).
	
	\item Finally, even when the upper bound
	(\ref{eq: upper bound}) is uniquely defined, we demonstrate that it is provably best unless $P=NP$
	(Theorem~\ref{theorem 3}).
\end{itemize}

Complementing these complexity-theoretic results, we also identify conditions under which the SLR-based bounds are guaranteed to provide a good approximation of the bilevel optimal value $z^*_{\, BP}$. Specifically we show that, under standard regularity assumptions, the gap between the SLR-based lower and upper bounds decreases at least linearly with the distance between the leader's and the follower's objective vectors (Theorem~\ref{theorem: approximation}). That is, both bounds converge to $z^*_{\, BP}$ as the two objectives become aligned. Notably, this approximation result also applies to the pure integer case ($n_1 = m_1 = 0$), in which~[\textbf{BP}] reduces to an \textit{integer bilevel linear program} (IBLP).

We next extend our complexity analysis to the pure integer case. In contrast to BLPs, IBLPs constitute a standard class of bilevel problems for which decomposition and cutting-plane-based methods are extensively used; see, e.g., \cite{Caprara2016,DeNegre2011}. However,  analyzing the quality of the SLR-based bounds for IBLPs requires complexity arguments at the second level of the polynomial hierarchy; see, e.g., \cite{Jeroslow1985} and Section \ref{sec: miblp} for further details. Moreover, in addition to the standard SLR-based lower bound $z^*_{\,SLR}$, we analyze a potentially weaker lower bound obtained from the linear programming relaxation of [\textbf{SLR}]. 

Let an \textit{MILP oracle} denote an oracle that, given any mixed-integer linear program~(MILP) of polynomial encoding size, returns an optimal solution. Assuming that the polynomial hierarchy does not collapse, we make the following contributions for the class of min-max~IBLPs:
\begin{itemize}
	\item We show that the linear programming relaxation of [\textbf{SLR}] provides a provably best \textit{polynomial-time
	computable} lower bound  (Theorem~\ref{theorem 4}).
	
	\item We show that the standard lower bound
	$z^*_{\,SLR}$ is provably best among lower bounds that are computable by polynomial-time algorithms with access to an MILP oracle (Theorem~\ref{theorem 5}).
	
	\item Finally, we establish the analogous result for the associated SLR-based upper bound~(\ref{eq: upper bound}), even when it is uniquely defined (Theorem~\ref{theorem 6}).
\end{itemize}
A summary of our key complexity-theoretic results is provided in~Table~\ref{tab:summary}.

Taken together, our results establish that, within their natural computational complexity regimes, the standard SLR-based lower and upper bounds admit no systematic improvement over broad classes of MIBLPs. In particular, our results for min-max IBLPs imply that, under standard complexity-theoretic assumptions, no generic decomposition or cutting-plane-based framework can be expected to systematically strengthen the standard SLR-based bounds within a polynomial number of~iterations. This result, however, does not rule out the existence of stronger bounds for particular instances or more restrictive classes of MIBLPs.  

\begin{table}[t] 
	\centering \doublespacing
	
	\small 
	\begin{tabular}{llll} 
		\toprule 
		Problem class & Bound & Result & Assumption \\ 
		\midrule 
		\multirow{3}{*}{Min-max BLPs}
		& SLR-based lower bound 
			& \multirow{3}{*}{$\left.\vcenter{\hbox{\rule{0pt}{3.5em}}}\right\}$ Provably best} 
		& $P \neq NP$ \\ 
		
		& KKT-relaxation lower bound 
		& 
		& $P \neq NP$ \\ 
		
		& SLR-based upper bound 
		& 
		& $P \neq NP$; uniqueness\\ 
		\hline
		\multirow{3}{*}{Min-max IBLPs}
		& LP relaxation of SLR 
		& \multirow{3}{*}{$\left.\vcenter{\hbox{\rule{0pt}{3.5em}}}\right\}$ Provably best}  
		& $P \neq NP$ \\ 
		
		& SLR-based lower bound 
		&
		& $\Delta_2^P\neq\Sigma_2^P$ \\ 
		
		& SLR-based upper bound 
		&
		& $\Delta_2^P\neq\Sigma_2^P$; uniqueness \\ 
		\bottomrule 
	\end{tabular} 
	
	\medskip 
	\footnotesize 
	\caption{Summary of the main complexity-theoretic results.  We use the standard complexity classes
		$\Sigma_2^P=NP^{NP}$, the second level of the polynomial hierarchy, and
		$\Delta_2^P=P^{NP}$, the class of problems solvable in deterministic polynomial time with access to an $NP$ oracle. Furthermore, uniqueness refers to the leader-optimal solution of [\textbf{SLR}] and the associated follower's problem (\ref{eq: follower's optimal response}). } 
	\label{tab:summary} 
\end{table}

The remainder of the paper is organized as follows. Sections~\ref{subsec: lower bounds} and~\ref{subsec: upper bounds} investigate the SLR-based lower and upper bounds for pure continuous BLPs, respectively. Section \ref{subsec: approximation} analyzes the quality of the SLR-based bounds when the decision-makers' objective functions are well aligned. In Section~\ref{subsec: poly lower bounds}, we analyze the SLR-based polynomial-time computable lower bound for IBLPs. Sections~\ref{subsec: np lower bounds} and~\ref{subsec: np upper bounds} investigate MILP-oracle computable lower and upper bounds for IBLPs, respectively. Finally, Section~\ref{sec: conclusions} concludes the paper and outlines directions for future research.

\textbf{Notation.} We use $\mathbb{R}_+$, $\mathbb{Z}_+$, and $\mathbb{Q}$ to denote the sets of nonnegative real numbers, nonnegative integers, and rational numbers, respectively. For any positive integer $k$, let $[k]:= \{1,\ldots,k\}$.
Vectors and matrices are denoted by boldface letters, with $\mathbf{1}$ representing the all-ones~vector of appropriate dimension. Finally, $\Vert \cdot \Vert$ denotes an arbitrary norm, and $\Vert \cdot \Vert_*$ its dual norm.

\section{Bilevel Linear Programs} \label{sec: blp}
The most well-studied class of MIBLPs are continuous \textit{bilevel linear programs} (BLPs), where the leader and the follower solve linear programs; see, e.g., \cite{Audet1997, Hansen1992} and the survey in~\cite{Kleinert2021}. Formally, BLPs correspond to [\textbf{BP}] with $n_2 = m_2 = 0$, i.e., 
\begin{subequations} \label{BLP} 
	\begin{align} [\textbf{BLP}]:\quad z^*_{\,BLP} := \min_{\mathbf{x},\mathbf{y}^*}\;& \mathbf{a}^\top \mathbf{x} + \mathbf{d}^\top \mathbf{y}^* \\ \text{s.t. }  & \mathbf{x} \in X^{c}  \label{cons: leader continuous} \\
		&\mathbf{y}^* \in \argmin_{\,\mathbf{y} \in Y^{c}(\mathbf{x})} \, \mathbf{g}^\top \mathbf{y}, \label{cons: follower continuous} \end{align} 
	\end{subequations}	
where $X^c := \{\mathbf{x} \in \mathbb{R}_+^{n_1}: \mathbf{H} \mathbf{x} \leq \mathbf{h} \}$ and $Y^{c}(\mathbf{x}) := \{\mathbf{y} \in \mathbb{R}_+^{m_1}: \mathbf{L} \mathbf{x} + \mathbf{F} \mathbf{y} \leq \mathbf{f}\}$. 

We make the following standard assumption (see, e.g., \cite{Audet1997, Kleinert2021}):
\begin{itemize}
	\item[\textbf{A1.}] The leader's feasible set $X^c$ is nonempty and bounded, and the follower's feasible set~$Y^c(\mathbf{x})$ is nonempty and bounded for all $\mathbf{x} \in X^c$.
\end{itemize}		
In particular, we note that, under Assumption~\textbf{A1}, both [\textbf{BLP}] and its single-level relaxation admit finite optimal solutions. 

It is known that [\textbf{BLP}] is strongly~$NP$-hard, even in the min-max case where $\mathbf{g} = -\mathbf{d}$; see, e.g.,~\cite{Hansen1992}. In contrast, when $n_2 = m_2 = 0$, the single-level relaxation [\textbf{SLR}] and the follower's problem in~(\ref{eq: follower's optimal response}) reduce
to linear programs. Therefore, both the lower bound $z^*_{\,SLR}$ and the
upper bound~$\hat{z}_{\,U}$ defined in~(\ref{eq: upper bound}) can be
computed in polynomial time. Given the computational complexity gap between solving [\textbf{BLP}] and computing the associated SLR-based bounds, it is therefore natural to investigate whether~$z^*_{\,SLR}$ or $\hat{z}_{\,U}$ admit  strict and uniform polynomial-time computable improvements in the sense of the research question~[\textbf{RQ}]. 

\subsection{Lower Bounds} \label{subsec: lower bounds}

Let us denote by $\mathcal{C}$ a class of BLPs of the form [\textbf{BLP}]. First, we analyze potential polynomial-time improvements of the lower bound $z^*_{\,SLR}(I)$ for all instances $I \in \mathcal{C}$ where $z^*_{\,SLR}(I)$ is not tight. That is, we consider the following decision problem:
\begin{itemize}
	\item[$ $] [\textbf{L-D}]: Given a class $\mathcal{C}$ of BLPs satisfying Assumption \textbf{A1}, does there exist a polynomial-time computable bound~$\beta_{\,L}$ such~that
	\begin{equation} \nonumber
		z^*_{\,SLR}(I) \leq \beta_{\,L}(I) \leq z^*_{\,BLP}(I) \quad \forall I \in \mathcal{C},
	\end{equation}
with strict inequality $\beta_{\,L}(I) > z^*_{\,SLR}(I)$ whenever $z^*_{\,SLR}(I) < z_{\,BLP}^*(I)$?	
\end{itemize}
The dependence on $I$ in [\textbf{L-D}] is omitted whenever clear from context. 
Furthermore, if the answer to~[\textbf{L-D}] is negative, then we refer to the corresponding lower bound $z^*_{\,SLR}$ as \textit{provably best}. The following result holds. 

\begin{theorem} \label{theorem 1}
	Unless $P = NP$, $z^*_{\,SLR}$ is provably best in the sense of \upshape [\textbf{L-D}], \itshape even when $\mathcal{C}$ is restricted to the class of min-max BLPs satisfying Assumption~\textbf{A1}.
	\begin{proof}
	Assume to the contrary that the answer to [\textbf{L-D}] is positive. That is, there exists  a polynomial-time computable bound $\beta_{\,L}$ such that $z^*_{\,SLR} < \beta_{\,L} \leq~z^*_{\,BLP}$ whenever~$z^*_{\,SLR} < z^*_{\,BLP}$, and $z^*_{\,SLR} = \beta_{\,L} = z^*_{\,BLP}$, otherwise. Since $z^*_{\,SLR}$ and~$\beta_{\,L}$ are polynomial-time computable, and $z^*_{\,SLR} = \beta_{\,L}$ if and only if $z^*_{\,SLR} = z^*_{\,BLP}$, it follows that the equality $z^*_{\,SLR} = z^*_{\,BLP}$ can be verified in polynomial time. 
		
		To derive a contradiction, consider an instance of 3-\textsc{SAT} given by a Boolean formula 
		\[
		\varphi := C_1 \wedge C_2 \wedge \cdots \wedge C_m
		\]
		over $n$ variables. Each clause $C_j$,
		$j \in [m]$, contains exactly three literals, where a literal
		is either $x_i$ or~$\neg x_i$ for some $i \in [n]$. The problem of determining whether $\varphi$ admits a satisfying assignment is known to be $NP$-complete \cite{Garey1979}. 
		
		Next, for each clause $C_j$, we define the index sets 
		\[
		P_j := \{ i : \text{literal } x_i \text{ appears in } C_j \} \; \mbox{ and } \;
		N_j := \{ i : \text{literal } \neg x_i \text{ appears in } C_j \},
		\]	
        and introduce the following associated min-max problem:
	\begin{subequations} \label{3 SAT min-max}
		\begin{align}
			z^*_{\,BLP} = \min_{\mathbf{x} } & \max_{\mathbf{y}} \Big\{\sum_{i = 1}^n y_i: \; \mathbf{0} \leq \mathbf{y} \leq \mathbf{x}, \; \mathbf{y} \leq \mathbf{1} - \mathbf{x} \Big\} \\
			\text{s.t. } &
			\sum_{i\in P_j} x_i + \sum_{i\in N_j} (1-x_i) \geq 1 \quad \forall j \in [m] \label{cons: 3 SAT min-max 1}\\
			& \mathbf{\mathbf{x}} \in [0,1]^n. \label{cons: 3 SAT min-max 2}
		\end{align}
	\end{subequations}
	It is rather easy to verify that Assumption \textbf{A1} holds. Moreover, if $\varphi$ admits a satisfying assignment $\tilde{\mathbf{x}} \in \{0, 1\}^n$, then  
	\[z^*_{\,BLP} = \max_{\mathbf{y}} \Big\{\sum_{i = 1}^n y_i: \; \mathbf{0} \leq \mathbf{y} \leq \tilde{\mathbf{x}}, \; \mathbf{y} \leq \mathbf{1} - \tilde{\mathbf{x}} \Big\} = \sum_{i = 1}^n \min\{\tilde{x}_i, 1 - \tilde{x}_i\} = 0.\]
	Otherwise, any feasible $\mathbf{x}$ in (\ref{3 SAT min-max}) has at least one fractional component, and hence $z^*_{\,BLP} > 0$. 
	
	On the other hand, the single-level relaxation of (\ref{3 SAT min-max}) is obtained by replacing the min-max problem~(\ref{3 SAT min-max}) with the min-min problem. By setting $x_i = \frac{1}{2}$ and $y_i = 0$, $i \in [n]$, we conclude $z^*_{\,SLR} = 0$. Hence, unless $P = NP$, the equality $z^*_{\,SLR} = z^*_{\,BLP}$ cannot be verified in polynomial time, and the result~follows.
	\end{proof}
\end{theorem}

As outlined in Section~\ref{sec: intro contributions}, [\textbf{BLP}] can also be reformulated as a single-level MILP by replacing the follower's problem in~(\ref{cons: follower continuous}) with its KKT optimality conditions and linearizing the resulting complementary slackness constraints using binary variables; see, e.g.,~\cite{Audet1997}. The resulting MILP reformulation of [\textbf{BLP}] reads as \begin{subequations} \label{MILP} \begin{align}
		\min_{\mathbf{x}, \mathbf{y}, \boldsymbol{\lambda}, \boldsymbol{\nu}, \mathbf{u}, \mathbf{v}} \; &  \mathbf{a}^{\top}\mathbf{x} + \mathbf{d}^{\top}\mathbf{y}  \\ 
		\text{s.t. } & \mathbf{x} \in X^c \label{cons: leader feasibility} \\ 
		& \mathbf{g} + \mathbf{F}^\top \boldsymbol{\lambda} - \boldsymbol{\nu} = \mathbf{0} \label{cons: stationarity}\\
		& \mathbf{0} \leq \boldsymbol{\lambda} \leq M \mathbf{u} \label{cons: big M 1 dual}\\
		& \mathbf{0} \leq \mathbf{f} - \mathbf{L} \mathbf{x} - \mathbf{F} \mathbf{y} \leq M (\mathbf{1} - \mathbf{u}) \label{cons: big M 1 primal} \\
		& \mathbf{0} \leq \boldsymbol{\nu} \leq M \mathbf{v} \label{cons: big M 2 dual}\\
		& \mathbf{0} \leq \mathbf{y} \leq M (\mathbf{1} - \mathbf{v}) \label{cons: big M 2 primal} \\
		& \mathbf{u} \in \{0, 1\}^{q}, \quad \mathbf{v} \in \{0, 1\}^{m_1},
\end{align} \end{subequations}
where $M > 0$ is a sufficiently large constant. We then consider a linear programming relaxation of (\ref{MILP}) given by:
 \begin{subequations} \label{KKT-LP} \begin{align}
 		z^*_{\,KKT} := \min_{\mathbf{x}, \mathbf{y}, \boldsymbol{\lambda}, \boldsymbol{\nu}, \mathbf{u}, \mathbf{v}} \; & \ \mathbf{a}^{\top}\mathbf{x} + \mathbf{d}^{\top}\mathbf{y}  \\ 
 		\text{s.t. } & \text{(\ref{cons: leader feasibility})--(\ref{cons: big M 2 primal}),} \\
 		& \mathbf{u} \in [0, 1]^{q}, \quad \mathbf{v} \in [0, 1]^{m_1}.
 \end{align} \end{subequations}

Assuming that $M$ is chosen so that the MILP reformulation (\ref{MILP}) is exact, the optimal value $z^*_{\,KKT}$ provides another valid lower bound for [\textbf{BLP}], i.e.,
$z^*_{\,KKT}\leq z^*_{\,BLP}$. Furthermore, a sufficiently large valid $M$ with polynomial encoding size exists; see, e.g., \cite{Buchheim2023}. The following result~holds. 

\begin{corollary} \label{corollary 1}
	Assume that $M$ is chosen so that \upshape (\ref{MILP}) \itshape is an exact reformulation of \upshape [\textbf{BLP}]. \itshape Then, under the assumptions of Theorem \ref{theorem 1}, the lower bound $z^*_{\,KKT}$ is provably best.
	\begin{proof}
	Since the KKT-based linear programming relaxation (\ref{KKT-LP}) retains all constraints
	of [\textbf{SLR}] and $M$ is valid, we conclude that
	\[
	z^*_{\,SLR}(I)\leq z^*_{\,KKT}(I)\leq z^*_{\,BLP}(I)
	\qquad \forall I\in\mathcal C.
	\]
	
	Suppose that $z^*_{\,KKT}$ admits a polynomial-time computable bound $\beta_{\,L}$ such that $z^*_{\,KKT} < \beta_{\,L} \leq~z^*_{\,BLP}$ whenever~$z^*_{\,KKT} < z^*_{\,BLP}$, and $z^*_{\,KKT} = \beta_{\,L} = z^*_{\,BLP}$, otherwise. Then, whenever
	$z^*_{\,SLR} < z^*_{\,BLP}$, either
	$z^*_{\,SLR} < z^*_{\,KKT} \leq z^*_{\,BLP}$ or
	$z^*_{\,SLR} = z^*_{\,KKT} < z^*_{\,BLP}$.
	In the former case, $\beta_{\,L}' := z^*_{\,KKT}$ is polynomial-time computable and strictly improves $z^*_{\,SLR}$. In the latter case, we set $\beta'_{\,L} := \beta_{\,L} > z^*_{\,SLR}$. Thus, $\beta'_{\,L}$ is a polynomial-time computable strict and uniform improvement
	of~$z^*_{\,SLR}$, contradicting Theorem~\ref{theorem 1}.	
\end{proof}
\end{corollary}


\subsection{Upper Bounds} \label{subsec: upper bounds}

Importantly, since both [\textbf{SLR}] and the follower's problem in~(\ref{eq: follower's optimal response}) may admit multiple optimal solutions, the resulting upper bound $\hat{z}_{\,U}$ defined by equation (\ref{eq: upper bound}) is generally not unique. In this regard, we first define the tightest upper~bound 
\begin{equation} \label{eq: tightest upper bound}
\hat{z}^*_{\,U}
:=
\min_{\hat{\mathbf{x}}, \hat{\mathbf{y}}}
\left\{
\mathbf{a}^{\top}\hat{\mathbf{x}}
+
\mathbf{d}^{\top}\hat{\mathbf{y}}: \;
\hat{\mathbf{x}}\in X^*_{\,SLR}, \; \hat{\mathbf{y}} \in \argmin_{\mathbf{y}} \big\{ \mathbf{g}^\top \mathbf{y}:\; \mathbf{y} \in Y(\hat{\mathbf{x}}) \big\}, 
\right\}.
\end{equation}
where $X^*_{\,SLR}$ denotes the set of leader-optimal solutions of~[\textbf{SLR}].
In particular, $\hat{z}^*_{\,U} \leq \hat{z}_{\,U}$ for any upper bound $\hat{z}_{\,U}$ defined by equation (\ref{eq: upper bound}). 
The following results show that, unlike the lower bound~$z^*_{\,SLR}$, computing $\hat{z}^*_{\,U}$ is $NP$-hard, even for  min-max BLPs.

\begin{lemma} \label{lemma 1}
	Consider a 0-1 integer linear program (ILP) given by:
	\begin{equation} \label{ILP}
		z^*_{\,ILP} := \min \Big\{\mathbf{c}^\top \mathbf{x}: \; \mathbf{A}\mathbf{x} \leq \mathbf{b}, \; \mathbf{x} \in \{0, 1\}^n \Big\},
	\end{equation}
	and its linear programming relaxation
	\begin{equation} \label{LP relaxation}
		z^*_{\,LP} := \min \Big\{\mathbf{c}^\top \mathbf{x}: \; \mathbf{A}\mathbf{x} \leq \mathbf{b}, \; \mathbf{x} \in [0, 1]^n \Big\},
	\end{equation}
	where $(\mathbf{A}, \mathbf{b}, \mathbf{c})$ are rational. Then, deciding whether $z^*_{\,LP} = z^*_{\,ILP}$ is $NP$-hard. 
	\begin{proof}
		Similar to the proof of Theorem \ref{theorem 1}, we consider an instance of 3-SAT given by a Boolean formula~$\varphi = C_1 \wedge C_2 \wedge \cdots \wedge C_m$, where $P_j$ and $N_j$ denote the sets of positive and negative literals in clause $C_j$, $j \in [m]$, respectively. We consider the following $0$--$1$ ILP:
		\begin{subequations} \label{3-SAT}
			\begin{align}
				z^*_{\,ILP} = \min_{\mathbf{x}, t} & \; t \\
				\text{s.t. } &
				\sum_{i\in P_j} x_i + \sum_{i\in N_j} (1-x_i) \geq 1 - t\quad \forall j \in [m] \\
				& \mathbf{x} \in \{0,1\}^n, \; t \in \{0, 1\}. 
			\end{align}
		\end{subequations}
		Notably, $z^*_{\,ILP} = 0$ if and only if $\varphi$ admits a satisfying assignment. Moreover, the LP relaxation of~(\ref{3-SAT}) has optimal value $z^*_{\,LP} = 0$, for example, by setting $x_i = \tfrac{1}{2}$ for each $i \in [n]$. Since 3-SAT is $NP$-complete, deciding whether $z^*_{\,LP}=z^*_{\,ILP}$ is $NP$-hard, and the result~follows.  
	\end{proof}
\end{lemma}

\begin{theorem} \label{theorem 2}
Computing $\hat{z}^*_{\,U}$ for \upshape [\textbf{BLP}] \itshape is $NP$-hard, even when $\mathbf{g} = -\mathbf{d}$.
\begin{proof}
Consider the 0-1 ILP given by (\ref{ILP}) and the following associated min-max problem:
	\begin{subequations} \label{ILP min-max}
	\begin{align}
		z^*_{\,BLP} = \min_{\mathbf{x}} \;& \left\{\mathbf{c}^\top \mathbf{x} +  \max_{\mathbf{y}} \Big\{\sum_{i = 1}^n y_i: \; \mathbf{0} \leq \mathbf{y} \leq \mathbf{x}, \; \mathbf{y} \leq \mathbf{1} - \mathbf{x} \Big\}\right\}  \\
		\text{s.t. } & \mathbf{A} \mathbf{x} \leq \mathbf{b} \label{cons: ILP min-max 1} \\
		& \mathbf{\mathbf{x}} \in [0,1]^n. \label{cons: ILP min-max 2} 
	\end{align}
\end{subequations}
First, we note that the single-level relaxation of (\ref{ILP min-max}) reads as
	\begin{subequations} \label{ILP SLR min-max}
	\begin{align}
		z^*_{\,SLR} = \min_{\mathbf{x}, \mathbf{y}} \;& \mathbf{c}^\top \mathbf{x} + \sum_{i = 1}^n y_i  \\
		\text{s.t. } & \text{(\ref{cons: ILP min-max 1})--(\ref{cons: ILP min-max 2})}, \\
		& \mathbf{0} \leq \mathbf{y} \leq \mathbf{x} \\ 
		& \mathbf{y} \leq \mathbf{1} - \mathbf{x}. 
	\end{align}
\end{subequations}
By construction, (\ref{ILP SLR min-max}) admits an optimal solution with $\mathbf{y} = \mathbf{0}$, and thus (\ref{ILP SLR min-max}) coincides with the LP relaxation (\ref{LP relaxation}), i.e., 
\begin{equation} \label{SLR min-max}
	z^*_{\,SLR} = z^*_{\,LP} := \min \Big\{\mathbf{c}^\top \mathbf{x}: \; \mathbf{A}\mathbf{x} \leq \mathbf{b}, \; \mathbf{x} \in [0, 1]^n \Big\}.
\end{equation}

Let $X^*_{\,SLR}$ denote the set of optimal solutions of (\ref{SLR min-max}). Then, the tightest upper bound~(\ref{eq: tightest upper bound}) can be expressed as
\begin{equation} \label{eq: tightest upper bound min-max}
	\begin{aligned}
	\hat{z}^*_{\,U} 
	:&=
	\min_{\hat{\mathbf{x}}}
	\Big\{
	\mathbf{c}^{\top}\hat{\mathbf{x}}
	+
	\sum_{i = 1}^n \min\{\hat{x}_i, 1 - \hat{x}_i\}: \;
	\hat{\mathbf{x}} \in X^*_{\,SLR}  
	\Big\} \\ & = z^*_{\,LP} + \min_{\hat{\mathbf{x}}}
	\Big\{
	\sum_{i = 1}^n \min\{\hat{x}_i, 1 - \hat{x}_i\}: \;
	\hat{\mathbf{x}} \in X^*_{\,SLR}  
	\Big\},
	\end{aligned}
\end{equation}
where we additionally use the fact that $\mathbf{c}^{\top} \hat{\mathbf{x}} = z^*_{\,LP}$ for any $\hat{\mathbf{x}} \in X^*_{\,SLR}$.  
From (\ref{eq: tightest upper bound min-max}), we observe that $\hat{z}^*_{\,U} = z^*_{\,LP}$
if and only if  $z^*_{\,ILP} = z^*_{\,LP}$, or equivalently, there exists an optimal solution $\hat{\mathbf{x}} \in X^*_{\,SLR}$ such that $\hat{\mathbf{x}} \in \{0,1\}^n$.
By Lemma~\ref{lemma 1}, deciding whether $z^*_{\,ILP} = z^*_{\,LP}$ is $NP$-hard. Since \(z^*_{\,LP}\) can be computed in polynomial time, we conclude that computing $\hat{z}^*_{\,U}$ is $NP$-hard, and the result follows.
\end{proof}
\end{theorem}

Similar to the analysis of the lower bound in Section~\ref{subsec: lower bounds}, we introduce the following decision problem:
\begin{itemize}
	\item[$ $] [\textbf{U-D}]: Given a class $\mathcal{C}$ of BLPs  satisfying Assumption~\textbf{A1}, does there exist a polynomial-time computable bound~$\beta_{\,U}$ such~that
	\begin{equation} \nonumber
		 z^*_{\,BLP}(I) \leq \beta_{\,U}(I) \leq \hat{z}^*_{\,U}(I) \quad \forall I \in \mathcal{C},
	\end{equation}
	with strict inequality $\beta_{\,U}(I) < \hat{z}^*_{\,U}(I)$ whenever $z^*_{\,BLP}(I) < \hat{z}_{\,U}^*(I)$?	
\end{itemize}
If the answer to [\textbf{U-D}] is negative, then we say that the corresponding upper bound $\hat{z}_{\,U}^*$ is provably~best.

However, based on the result of Theorem \ref{theorem 2}, we further restrict our attention to instances of [\textbf{BLP}] for which both the single-level relaxation [\textbf{SLR}] and the associated follower's problem in (\ref{eq: follower's optimal response}) admit unique optimal solutions. In this case, the upper bound in (\ref{eq: upper bound}) is uniquely determined, and therefore the tightest upper bound~(\ref{eq: tightest upper bound}) can be computed in polynomial time. Next, we show that $\hat{z}_{\,U}^*$ is provably best even for the class of min-max BLPs. The following preliminary result holds.
\begin{lemma} \label{lemma 2}Let
	\begin{equation} \label{eq: polytope}
		P:=\{\mathbf{u}\in[0,1]^n:\mathbf{A}\mathbf{u}\leq\mathbf{b}\}
	\end{equation}
	be a nonempty polytope, where
	$\mathbf{A}\in\{-1,0,1\}^{m\times n}$,
	$\mathbf{b}\in\mathbb{Z}^m$, and each row of $\mathbf{A}$ contains at most $q \geq 1$ nonzero entries. If
	$P\cap\{0,1\}^n=\emptyset$, then
	\[
	\min_{\mathbf{u}\in P} \,
	\sum_{i=1}^n \min\{u_i,1-u_i\}
	\ge
	q^{-n/2}.
	\]
 \begin{proof} The function
 	\[
 	\psi(\mathbf{u})
 	:=
 	\sum_{i=1}^n \min\{u_i,1-u_i\}
 	\]
 	is concave, and $P$ is nonempty and compact.  Hence, the minimum of $\psi$ over $P$ is attained
 	at a vertex~$\bar{\mathbf{u}}$ of $P$. The vertex $\bar{\mathbf{u}}$ is determined by $n$ linearly independent active
 	constraints, potentially including the bound constraints
 	$0\le u_i \le1$, $i \in [n]$, whose coefficient matrix
 	$\bar{\mathbf{A}}\in\{-1,0,1\}^{n\times n}$ is nonsingular. Each row of $\bar{\mathbf{A}}$ has entries in $\{-1,0,1\}$ and
 	Euclidean norm at most~$\sqrt q$. 
 	
 	Then, by Hadamard's inequality, we have
 	$
 	|\det(\bar{\mathbf{A}})|
 	\le
 	q^{n/2}.
 	$
 	Furthermore, by Cramer's rule, each coordinate of $\bar{\mathbf{u}}$ is rational with the denominator at most $q^{n/2}$. Since~$P\cap\{0,1\}^n=\emptyset$, the vertex $\bar{\mathbf{u}}$ is not binary. Thus, there exists an index $i \in [n]$ such that $0<\bar u_i<1$. For this~index,
 	\[
 	\min\{\bar u_i,1-\bar u_i\}
 	\ge
 	q^{-n/2}.
 	\]
 	Consequently,
 	\[
 	\sum_{i=1}^n \min\{\bar u_i,1-\bar u_i\}
 	\ge
 	q^{-n/2},
 	\]
 	which implies the result. \end{proof}
\end{lemma}

\begin{theorem} \label{theorem 3}
	Let $\mathcal{C}$ be restricted to the class of min-max BLPs satisfying Assumption \textbf{A1}, for which both the single-level relaxation~\upshape[\textbf{SLR}] \itshape and the associated follower's problem in \upshape (\ref{eq: follower's optimal response}) \itshape admit unique optimal solutions.
	Then, unless $P = NP$, the tightest upper bound $\hat{z}_{\,U}^*$ is provably best in the sense of \upshape [\textbf{U-D}]. \itshape  
	\begin{proof}
		Similar to the proof of Theorem \ref{theorem 1}, we consider an instance of 3-SAT given by a Boolean formula~$\varphi = C_1 \wedge C_2 \wedge \cdots \wedge C_m$, where $P_j$ and $N_j$ denote the sets of positive and negative literals in clause $C_j$, $j \in [m]$, respectively. We introduce the following associated min-max problem:
	\begin{subequations} \label{3 SAT min-max upper bound}
		\begin{align}
			z^*_{\,BLP} = \min_{\mathbf{x}, s } &\; M_1 s + \max_{\mathbf{y}} \Big\{M_2 \sum_{i = 1}^n y_i: \; \mathbf{0} \leq \mathbf{y} \leq \mathbf{x}, \; \mathbf{y} \leq \mathbf{1} - \mathbf{x} \Big\} \\
			\text{s.t. } &
			\sum_{i\in P_j} x_i + \sum_{i\in N_j} (1-x_i) \geq \tfrac{3-s}{2} \quad \forall j \in [m] \label{cons: 3 SAT min-max upper bound 1} \\
			& \tfrac{1 - s}{2} \leq x_i \leq \tfrac{1 + s}{2} \quad \forall i \in [n] \label{cons: 3 SAT min-max upper bound 2} \\
			& s \in [0, 1], \label{cons: 3 SAT min-max upper bound 3}
		\end{align}
	\end{subequations}	
	where $M_2 := 3^{n} > 0$ and $M_1 := \tfrac{M_2 n}{2} - 1 > 0$. Since $M_2 > 0$ and $s \in [0, 1]$, the follower's optimal solution in (\ref{3 SAT min-max upper bound}) is unique and given by $y^*_i = \min\{x_i, 1 - x_i\}$, $i \in [n]$. Furthermore, Assumption \textbf{A1} holds by construction.
	
	We now analyze the single-level relaxation of (\ref{3 SAT min-max upper bound}) defined as
		\begin{subequations} \label{3 SAT min-max SLR}
		\begin{align}
			z^*_{\,SLR} = \min_{\mathbf{x}, s, \mathbf{y}} &\; M_1 s + M_2 \sum_{i = 1}^n y_i \\
			\text{s.t. } & \text{(\ref{cons: 3 SAT min-max upper bound 1})--(\ref{cons: 3 SAT min-max upper bound 3})} 
			\\
			& \mathbf{0} \leq \mathbf{y} \leq \mathbf{x} \\ 
			& \mathbf{y} \leq \mathbf{1} - \mathbf{x}. 
		\end{align}
	\end{subequations}	
	It is rather easy to verify that $\hat{\mathbf{y}}^* = \mathbf{0}$, $\hat{s}^* = 0$ and $\hat{\mathbf{x}}^* = \tfrac{1}{2} \mathbf{1}$ is the unique optimal solution of (\ref{3 SAT min-max SLR}), with the associated optimal objective function value $z^*_{\,SLR} = 0$. Substituting $\hat{s}^* = 0$ and $\hat{\mathbf{x}}^* = \tfrac{1}{2} \mathbf{1}$ into the follower's problem in (\ref{3 SAT min-max upper bound}) yields the tightest upper bound
	\[\hat{z}^*_{\,U} = M_1 \hat{s}^* + M_2 \sum_{i = 1}^n \min\{\hat{x}^*_i, 1 - \hat{x}^*_i\} = \tfrac{M_2 n}{2}. \]
	
	First, assume that $\varphi$ admits a satisfying assignment $\tilde{\mathbf{x}} \in \{0, 1\}^n$. Then, by setting $\mathbf{x} = \tilde{\mathbf{x}}$ and~$s = 1$, we observe that the leader constraints (\ref{cons: 3 SAT min-max upper bound 1})--(\ref{cons: 3 SAT min-max upper bound 3}) are satisfied. Hence,
	\begin{equation} \nonumber
	z^*_{\,BLP} \leq M_1 + M_2 \sum_{i = 1}^n \min\{\tilde{x}_i, 1 - \tilde{x}_i\} = M_1 = \tfrac{M_2 n}{2} - 1 < \tfrac{M_2 n}{2} = \hat{z}^*_{\,U},
	\end{equation}
    where we use the definition of $M_1$ and the fact that $\tilde{\mathbf{x}} \in \{0, 1\}^n$. Thus, if the answer to 3-SAT is ``yes'', then $z^*_{\,BLP} < \hat{z}^*_{\,U}$.
    
    Next, suppose that $\varphi$ is unsatisfiable. If $s = 0$, then the unique feasible solution of the leader in~(\ref{3 SAT min-max upper bound}) is given by $\mathbf{x} = \tfrac{1}{2} \mathbf{1}$, yielding the objective function value $\tfrac{M_2 n}{2}$. 
    On the other hand, with~$s \in (0, 1]$, any feasible $\mathbf{x}$ satisfies
    \[x_i = \tfrac{1 - s}{2} + u_i s,\]
    where $u_i \in [0, 1]$ and $i \in [n]$; recall (\ref{cons: 3 SAT min-max upper bound 2}). In particular, $1 - x_i = \tfrac{1 - s}{2} + (1 - u_i) s$ and, for each $j \in [m]$, the left-hand side of constraints (\ref{cons: 3 SAT min-max upper bound 1}) can be expressed as
    \begin{equation} \nonumber
    	\sum_{i\in P_j} x_i + \sum_{i\in N_j} (1-x_i) 
    	= \tfrac{3}{2}(1 - s) + \Big(\sum_{i\in P_j}u_i + \sum_{i\in N_j} (1-u_i) \Big) s.
    \end{equation}
   Thus, (\ref{cons: 3 SAT min-max upper bound 1}) implies that 
   $\sum_{i\in P_j}u_i + \sum_{i\in N_j} (1-u_i) \geq 1$. 
   Furthermore,
   \[\sum_{i = 1}^n \min \{x_i, 1 - x_i\} = \tfrac{1 - s}{2} n + s \sum_{i = 1}^n \min\{u_i, 1 - u_i\}.\] 
   
   As a result, by combining the cases $s = 0$ and $s \in (0, 1]$, we observe that (\ref{3 SAT min-max upper bound}) reduces to 
   \begin{subequations} \label{3 SAT min-max upper bound 2}
   	\begin{align}
        z^*_{\,BLP} = \min_{\mathbf{u}, s} &\; \left\{M_1 s + M_2 \Big(\tfrac{1 - s}{2} n + s\sum_{i = 1}^n \min\{u_i, 1 - u_i\}  \Big)\right\} \\
   		\text{s.t. } &
   		\sum_{i\in P_j} u_i + \sum_{i\in N_j} (1-u_i) \geq 1 \quad \forall j \in [m] \label{cons: 3 SAT min-max upper bound 2 1}  \\
   		& u_i \in [0, 1] \quad \forall i \in [n] \textbf{\label{cons: 3 SAT min-max upper bound 2 2}} \\
   		& s \in [0, 1]. \label{cons: 3 SAT min-max upper bound 2 3}
   	\end{align}
   \end{subequations}
   In particular, with $s = 0$, one may choose any feasible $\mathbf{u}$, for example, $\mathbf{u} = \tfrac{1}{2} \mathbf{1}$.
   Since \(\varphi\) is unsatisfiable, the feasible region of
   (\ref{3 SAT min-max upper bound 2}) contains no binary vector \(\mathbf{u}\in\{0,1\}^n\). Hence,
   \[
   \delta
   :=
   \min_{\mathbf{u}}
   \Big\{
   \sum_{i=1}^n \min\{u_i,1-u_i\}
   : \;
   \text{(\ref{cons: 3 SAT min-max upper bound 2 1})--(\ref{cons: 3 SAT min-max upper bound 2 2}) hold}
   \Big\}
   >0,
   \]
   and furthermore 
   \begin{equation} \nonumber
   	 z^*_{\,BLP} = \min_{s \in [0, 1]} \Big\{(M_1 - \tfrac{M_2 n}{2} + M_2\delta) s  + \tfrac{M_2 n}{2} \Big\} = \min_{s \in [0, 1]} \Big\{(M_2 \delta - 1) s + \tfrac{M_2 n}{2} \Big\}.
   \end{equation}
   
   Notably, constraints (\ref{cons: 3 SAT min-max upper bound 2 1})--(\ref{cons: 3 SAT min-max upper bound 2 2}) define a nonempty compact polytope of the form \eqref{eq: polytope}, whose constraint matrix has at most $q=3$ nonzero entries in each row.
   Consequently, by Lemma \ref{lemma 2}, $\delta \geq 3^{-n/2}$, and thus $M_2 \delta \geq 3^{n/2} > 1$. We conclude that $z^*_{\,BLP} = \tfrac{M_2 n}{2} = \hat{z}^*_{\,U}$ if and only if $\varphi$ is unsatisfiable. This contradicts the existence of a polynomial-time computable upper bound $\beta_{\,U}$ that strictly dominates~$\hat{z}^*_{\,U}$, and the result follows.
    \end{proof}
\end{theorem}  

\subsection{Approximation guarantees for SLR-based bounds} \label{subsec: approximation}
Importantly, Theorems~\ref{theorem 1} and~\ref{theorem 3}
apply to a general min-max formulation of~[\textbf{BLP}]. In contrast, whenever~[\textbf{BLP}] itself is polynomially solvable, $z^*_{\,BLP}$ provides the required polynomial-time improvement, and therefore the answers to~[\textbf{L-D}] and [\textbf{U-D}] become positive. This is the case, for example, when either the number of follower variables or constraints in~(\ref{cons: follower continuous}) is fixed; see, e.g., \cite{Deng1998,Ketkov2026}. Moreover, the single-level relaxation [\textbf{SLR}] is known to be exact when~[\textbf{BLP}] is a \textit{min-min problem}, i.e., $\mathbf{g} = \mathbf{d}$. In this case, $z^*_{\,SLR} = z^*_{\,BLP} = \hat{z}^*_{\, U}$ and [\textbf{BLP}] reduces to a linear program; recall the definition of~[\textbf{SLR}] and the tightest upper~bound~(\ref{eq: tightest upper bound}). 

This observation raises the question of whether the SLR-based bounds $z^*_{\,SLR}$ and $\hat{z}^*_{\, U}$ remain tight when the leader's and the follower's objective functions are  well aligned, i.e., $\| \mathbf{d} - \mathbf{g}\| \leq \varepsilon$ for some small $\varepsilon > 0$. To this end, we establish the following general approximation~result. 

\begin{theorem}\label{theorem: approximation}
	Suppose that Assumption~\textbf{A1} holds and
	$\|\mathbf{d}-\mathbf{g}\| \leq \varepsilon$. Define
	\begin{equation} \label{eq: diameter}
		D := \max_{\mathbf{x}\in X^c} \,
		\max_{\mathbf{y},\mathbf{y}'\in Y^c(\mathbf{x})}
		\|\mathbf{y}-\mathbf{y}'\|_*.
	\end{equation}
	Then, $D<\infty$ and
	$0\leq \hat{z}^*_{\,U}-z^*_{\,SLR}\leq \varepsilon D$. 
\end{theorem}

\begin{proof}
	Under Assumption~\textbf{A1}, $X^c$ and $Y^c(\mathbf{x})$ for every
	$\mathbf{x}\in X^c$ are bounded polyhedrons, and hence the constant~$D$ defined by equation (\ref{eq: diameter}) is finite. 
	Let $\hat{\mathbf{x}}^*$ be any leader-optimal solution of [\textbf{SLR}] with $n_2 = m_2 = 0$ and let
	\[
	\tilde{\mathbf{y}}^*\in
	\argmin_{\mathbf{y}\in Y^c(\hat{\mathbf{x}}^*)}
	\mathbf{d}^{\top}\mathbf{y} \, \; \text{ and } \, \; \hat{\mathbf{y}}^* \in
	\argmin_{\mathbf{y}\in Y^c(\hat{\mathbf{x}}^*)}
	\mathbf{g}^{\top}\mathbf{y}.
	\]
	By definition, we have
    $z^*_{\,SLR}
	=
	\mathbf{a}^{\top}\hat{\mathbf{x}}^*
	+\mathbf{d}^{\top}\tilde{\mathbf{y}}^*$ and $\hat{z}_{\,U} = \mathbf{a}^{\top}\hat{\mathbf{x}}^*
	+\mathbf{d}^{\top}\hat{\mathbf{y}}^*$, where $\hat{z}_{\,U}$ is the upper bound defined by equation (\ref{eq: upper bound}). 
	
	Therefore,
	\begin{equation} \label{eq: approximation}
		\hat z_{\,U}-z^*_{\,SLR}
		=
		\mathbf{d}^{\top}(\hat{\mathbf{y}}^*-\tilde{\mathbf{y}}^*) \leq
		(\mathbf{d}-\mathbf{g})^{\top}
		(\hat{\mathbf{y}}^*-\tilde{\mathbf{y}}^*) \leq
		\|\mathbf{d}-\mathbf{g}\| \,
		\|\hat{\mathbf{y}}^*-\tilde{\mathbf{y}}^*\|_* \leq
		\varepsilon D.
	\end{equation}
	Here, the first inequality follows from the
	optimality of $\hat{\mathbf{y}}^*$, i.e.,
	$\mathbf{g}^{\top}\hat{\mathbf{y}}^*
	\leq \mathbf{g}^{\top}\tilde{\mathbf{y}}^*$, while the second follows from the
	Hölder's inequality. Since~(\ref{eq: approximation}) holds for
	every associated SLR-based upper bound $\hat z_{\,U}$, it also holds for
	the tightest upper bound $\hat{z}^*_{\,U}$ defined by equation (\ref{eq: tightest upper bound}). Finally,
	$z^*_{\,SLR}\leq z^*_{\,BLP}\leq \hat{z}^*_{\,U}$,
	which implies the result.
\end{proof}

The constant $D$ defined by equation (\ref{eq: diameter}) depends on the chosen norm and may, in general, be difficult to compute.
In particular, for the Euclidean norm, computing $D$ involves norm maximization over a polyhedron, a problem that is $NP$-hard in general \cite{Bodlaender1990}. In contrast, when the primal norm $\Vert\cdot\Vert$ is the $\ell_1$-norm, so that the dual norm $\Vert\cdot\Vert_*$ is the $\ell_\infty$-norm, we~have
\[
D
=
\max_{\mathbf{x}\in X^c}
\max_{\mathbf{y},\mathbf{y}'\in Y^c(\mathbf{x})}
\Vert \mathbf{y}-\mathbf{y}' \Vert_\infty
=
\max_{i\in[m_1]}
\max_{\substack{\mathbf{x}\in X^c\\
		\mathbf{y},\mathbf{y}'\in Y^c(\mathbf{x})}}
|y_i-y_i'|.
\]
Hence, $D$ can be computed in polynomial time by solving $2m_1$ linear programs.
\color{black}

Overall, Theorem~\ref{theorem: approximation} shows that the gap between the SLR-based lower and upper bounds decreases at least linearly with the distance between the leader's and follower's objective vectors. In particular, as $\varepsilon \to 0$, both bounds converge to $z^*_{\,BLP}$, recovering their exactness in the min-min case with $\mathbf{g}=\mathbf{d}$. Moreover, 
the proof of Theorem~\ref{theorem: approximation} does not rely on the continuity of the leader or follower variables and therefore, under standard regularity assumptions, applies to the pure integer case ($n_1 = m_1 = 0$).
\section{Integer Bilevel Linear Programs} \label{sec: miblp}
In this section we consider pure integer bilevel linear programs (IBLPs), where the leader and the follower solve integer linear programs; see, e.g., \cite{Caprara2014, Caprara2016,DeNegre2011} and the survey in \cite{Kleinert2021}. In particular, IBLPs correspond to [\textbf{BP}] with $n_1 = m_1 = 0$, i.e.,
\begin{subequations}\label{IBLP}
	\begin{align}
		[\textbf{IBLP}]:\quad
		z^*_{\,IBLP}:=\min_{\mathbf{x},\mathbf{y}^*}\;&
		\mathbf{a}^\top\mathbf{x}+\mathbf{d}^\top\mathbf{y}^*\\
		\text{s.t. }
		&\mathbf{x}\in X^d,\label{cons: leader integer}\\
		&\mathbf{y}^*\in
		\argmin_{\,\mathbf{y}\in Y^d(\mathbf{x})} \,
		\mathbf{g}^\top\mathbf{y}.
		\label{cons: follower integer}
	\end{align}
\end{subequations}
where $X^d:=\{\mathbf{x}\in\mathbb{Z}_+^{n_2}:\mathbf{H}\mathbf{x}\le\mathbf{h}\}$
and $Y^d(\mathbf{x})
:=
\{\mathbf{y}\in\mathbb{Z}_+^{m_2}:
\mathbf{L}\mathbf{x}+\mathbf{F}\mathbf{y}\le\mathbf{f}\}$. 
Similar to the pure continuous case, we make the following standard assumption:
\begin{itemize}
	\item[\textbf{A1$'$.}] The leader's feasible set $X^d$ is nonempty and bounded, and the follower's feasible set~$Y^d(\mathbf{x})$ is nonempty and bounded for all $\mathbf{x} \in X^d$.
\end{itemize}	

It is well known that [\textbf{IBLP}] is $\Sigma^P_2$-hard; see, e.g., \cite{Caprara2013, Jeroslow1985}. In other words, this problem is located at the second level of the polynomial hierarchy and, unless $NP = \Sigma^P_2$, there is no way of formulating it as a single-level mixed-integer linear programming~(MILP) problem of polynomial size. In contrast, when $n_1 = m_1 = 0$, [\textbf{SLR}] and the follower's problem in~(\ref{eq: follower's optimal response}) reduce to integer linear programs. We therefore conclude that computing both the lower bound~$z^*_{\,SLR}$ and the upper bound $\hat{z}_{\,U}$ defined in (\ref{eq: upper bound}) is $NP$-hard~\cite{Garey1979}. This again motivates the question of whether the SLR-based bounds can be uniformly improved within the same computational complexity regime. 

In addition to the standard single-level relaxation [\textbf{SLR}], one may also analyze its linear programming relaxation given by:
 \begin{subequations} \label{SLR 2} 
	\begin{align} [\textbf{SLR}']:\quad \tilde{z}^{*}_{\,SLR} := \min_{\mathbf{x},\mathbf{y}}\;& \mathbf{a}^\top \mathbf{x} + \mathbf{d}^\top \mathbf{y}  \\ \text{s.t. } & \mathbf{H} \mathbf{x}  \leq \mathbf{h} \\
	& \mathbf{L} \mathbf{x} + \mathbf{F} \mathbf{y} \leq \mathbf{f} \\
	& \mathbf{x}, \mathbf{y} \geq \mathbf{0};
\end{align} \end{subequations}
see, e.g., \cite{Kleinert2021}. In particular, $\tilde{z}^{*}_{\,SLR}$ can be computed in polynomial time and provides a valid, albeit generally weaker, lower bound for $z^*_{\,IBLP}$, i.e., we have \[\tilde{z}^{*}_{\,SLR} \leq z^{*}_{\,SLR}  \leq z^*_{\,IBLP}.\] Meanwhile, an optimal solution of [\textbf{SLR}$'$] is not necessarily integer, and hence this solution cannot be used to obtain a valid upper bound for $z^*_{\,IBLP}$; recall (\ref{eq: upper bound}). 

\subsection{Polynomial-time Computable Lower Bounds} \label{subsec: poly lower bounds}
First, we demonstrate that the result similar to Theorem \ref{theorem 1} can be readily obtained for the continuous relaxation [\textbf{SLR}$'$]. Similar to [\textbf{L-D}], we introduce the following decision problem:
\begin{itemize}
	\item[$ $] [\textbf{L-D}$'$]: Given a class $\mathcal{C}$ of IBLPs satisfying Assumption \textbf{A1$'$}, does there exist a polynomial-time computable bound~$\beta_{\,L}$ such~that
	\begin{equation} \nonumber
		\tilde{z}^*_{\,SLR}(I) \leq \beta_{\,L}(I) \leq z^*_{\,IBLP}(I) \quad \forall I \in \mathcal{C},
	\end{equation}
	with strict inequality $\beta_{\,L}(I) > \tilde{z}^*_{\,SLR}(I)$ whenever $\tilde{z}^*_{\,SLR}(I) < z_{\,IBLP}^*(I)$?
\end{itemize}
The following result holds.
\begin{theorem} \label{theorem 4}
	Unless $P = NP$, the lower bound $\tilde{z}^*_{\,SLR}$ is provably best in the sense of \upshape [\textbf{L-D}$'$], \itshape even when $\mathcal{C}$ is restricted to the class of min-max IBLPs satisfying Assumption \textbf{A1$\,'$}.
	\begin{proof}
		Similar to the proof of Theorem \ref{theorem 1}, we consider an instance of 3-SAT given by a Boolean formula~$\varphi = C_1 \wedge C_2 \wedge \cdots \wedge C_m$, where $P_j$ and $N_j$ denote the sets of positive and negative literals in clause $C_j$, $j \in [m]$, respectively. We introduce the following associated min-max problem:
		\begin{subequations} \label{3 SAT min-max integer}
			\begin{align}
				z^*_{\,IBLP} = \min_{\mathbf{x}, t } & \;\left\{\ t + \max_{\mathbf{y} \in \{0, 1\}^n} \Big\{\sum_{i = 1}^n y_i: \; \mathbf{0} \leq \mathbf{y} \leq \mathbf{x}, \; \mathbf{y} \leq \mathbf{1} - \mathbf{x} \Big\}\right\} \\
				\text{s.t. } &
				\sum_{i\in P_j} x_i + \sum_{i\in N_j} (1-x_i) \geq 1 - t \quad \forall j \in [m] \\
				& \mathbf{\mathbf{x}} \in \{0,1\}^n, \; t \in \{0, 1\}. 
			\end{align}
		\end{subequations}
		
		If $\varphi$ admits a satisfying assignment $\tilde{\mathbf{x}} \in \{0, 1\}^n$, then setting $\mathbf{x} = \tilde{\mathbf{x}}$ and $t = 0$ yields $z^*_{\,IBLP} = 0$.   
		Otherwise, any optimal solution of (\ref{3 SAT min-max integer}) satisfies $t^* = 1$ and $\mathbf{y}^* = \mathbf{0}$, which yields $z^*_{\,IBLP} = 1$. Finally, the continuous relaxation of (\ref{3 SAT min-max integer}) has optimal value $\tilde{z}^*_{\,SLR} = 0$, for example, by setting $\mathbf{x} = \tfrac{1}{2} \mathbf{1}$ and~$t = 0$. Hence, unless $P = NP$, the equality $\tilde{z}^*_{\,SLR} = z^*_{\,IBLP}$ cannot be verified in polynomial time, and the result follows.
	\end{proof}
\end{theorem}
\subsection{MILP-Oracle Computable Lower Bounds} \label{subsec: np lower bounds}
In this section, we consider the integer single-level relaxation [\textbf{SLR}] corresponding to $n_1 = m_1=~0$. Since [\textbf{SLR}] is itself an integer linear program, we investigate whether~$z^*_{\,SLR}$ can be systematically improved by using polynomial-time algorithms with access to an \textit{MILP~oracle}:

\begin{definition} \label{definition 1}
\upshape A bound $\beta$ is \textit{MILP-oracle computable} if it can be computed by a polynomial-time algorithm with access to an oracle that optimally solves polynomial-size~MILPs. \hfill $\square$
\end{definition}

Definition~\ref{definition 1} is motivated by the fact that many exact algorithms for [\textbf{IBLP}] strengthen the single-level relaxation [\textbf{SLR}] by iteratively adding valid inequalities and repeatedly solving polynomial-size MILPs; see, e.g., \cite{Caprara2016, DeNegre2011}. We therefore investigate whether such iterative MILP-based procedures can systematically improve $z^*_{\,SLR}$. To this end, we introduce the following analogue of [\textbf{L-D}]:
\begin{itemize}
	\item[$ $] [\textbf{IL-D}]: Given a class $\mathcal{C}$ of IBLPs satisfying Assumption~\textbf{A1$'$}, does there exist an MILP-oracle computable bound~$\beta_{\,L}$, such~that
	\begin{equation} \nonumber
		z^*_{\,SLR}(I) \leq \beta_{\,L}(I) \leq z^*_{\,IBLP}(I) \quad \forall I \in \mathcal{C},
	\end{equation}
	with strict inequality $\beta_{\,L}(I) > z^*_{\,SLR}(I)$ whenever $z^*_{\,SLR}(I) < z_{\,IBLP}^*(I)$?
\end{itemize}

Let $\Delta_2^P = P^{NP}$ denote the class of problems solvable
in polynomial time with access to an $NP$ oracle; see, e.g., \cite{Arora2009}. Then, the following result holds.

\begin{theorem} \label{theorem 5}
	Unless $\Delta_2^P = \Sigma^P_2$, the lower bound $z^*_{\,SLR}$ is provably best in the sense of \upshape [\textbf{IL-D}], \itshape even when $\mathcal{C}$ is restricted to the class of min-max~IBLPs satisfying Assumption~\textbf{A1$\,'$}.
\begin{proof} First, we observe that $\beta_{\,L} = z^*_{\,SLR}$ if and only if
$z^*_{\,SLR} = z^*_{\,IBLP}$. In the following, we show that deciding whether
$z^*_{\,SLR} = z^*_{\,IBLP}$ is $\Sigma_2^P$-hard, and therefore the answer to~[\textbf{IL-D}] is negative.
Indeed, if the answer to [\textbf{IL-D}] were positive, then one could
decide whether $z^*_{\,SLR} = z^*_{\,IBLP}$ by comparing $\beta_{\,L}$ and $z^*_{\,SLR}$ that are MILP oracle-computable, which contradicts the assumption that $\Delta^P_2 \neq \Sigma_2^P$.

To establish that deciding whether
$z^*_{\,SLR} = z^*_{\,IBLP}$ is $\Sigma_2^P$-hard, we use a reduction from QSAT$_2$. Given a quantified Boolean formula
\[
\exists \mathbf{x}\in\{0,1\}^{n_x}\;
\forall \mathbf{y}\in\{0,1\}^{n_y}\;
\varphi(\mathbf{x},\mathbf{y}),
\]
where $\varphi(\mathbf{x},\mathbf{y}) = T_1(\mathbf{x},\mathbf{y}) \vee \cdots \vee T_m(\mathbf{x},\mathbf{y})$ is in
\textit{3-disjunctive normal form} (3-DNF), the question is whether the quantified formula is true.
This problem is known to be $\Sigma_2^P$-complete~\cite{Stockmeyer1976}.

For each term $T_j$, $j \in [m]$, let $P_j^x$ and $N_j^x$ (respectively, $P_j^y$ and $N_j^y$) denote the sets of existential~(universal) variables that appear as positive and negative literals in term $T_j$, respectively. We~consider the following instance of  [\textbf{IBLP}] associated with QSAT$_2$:

\begin{subequations} \label{QSAT min-max integer}
	\begin{align}
		z^*_{\,IBLP} = \min_{\mathbf{x} \in \{0,1\}^{n_x}} \max_{\mathbf{y}, v} \quad & v \\
		\text{s.t. } 
		& v \leq 
		\sum_{i \in P^x_j}(1-x_i)
		+\sum_{i \in N^x_j}x_i
		+\sum_{i \in P^y_j}(1-y_i)
		+\sum_{i \in N^y_j}y_i
		\quad \forall j \in [m]
		\label{cons: QSAT min-max integer 2} \\
		& v \in \{0,1\}, \; \mathbf{y} \in \{0,1\}^{n_y}. 
		\label{cons: QSAT min-max integer 1}
	\end{align}
\end{subequations}

Suppose that QSAT$_2$ admits a ``yes'' instance. Then, there exists
$\tilde{\mathbf{x}} \in \{0,1\}^{n_x}$ such that
$\varphi(\tilde{\mathbf{x}},\mathbf{y})$ is satisfied for every
$\mathbf{y} \in \{0,1\}^{n_y}$. Hence, for every
$\mathbf{y} \in \{0,1\}^{n_y}$, at least one term $T_j$ is satisfied. For this term, the
right-hand side of the corresponding constraint
(\ref{cons: QSAT min-max integer 2}) is equal to zero. Thus, the follower's optimal solution in
(\ref{QSAT min-max integer}) satisfies $v^*=0$, and consequently
$z^*_{\,IBLP}=0$.

Conversely, suppose that QSAT$_2$ admits a ``no'' instance. Then, for every
$\mathbf{x}\in\{0,1\}^{n_x}$, there exists
$\tilde{\mathbf{y}}\in\{0,1\}^{n_y}$ such that
$\varphi(\mathbf{x},\tilde{\mathbf{y}})$ is not satisfied. Equivalently,
every term $T_j$, $j\in[m]$, contains at least one literal that is false.
By setting $\mathbf{y}=\tilde{\mathbf{y}}$, the right-hand side of every
constraint \eqref{cons: QSAT min-max integer 2} is therefore at least one,
and hence the follower can set $v=1$. Thus, the follower's optimal
objective function value in (\ref{QSAT min-max integer}) is equal to one.
Since this holds for every leader decision $\mathbf{x} \in\{0,1\}^{n_x}$, we obtain
$z^*_{\,IBLP}=1$.

Finally, it is rather straightforward to verify that $z^*_{SLR} = 0$ for the single-level relaxation of (\ref{QSAT min-max integer}). Thus, $z^*_{IBLP} = z^*_{SLR}$ if and only if QSAT$_2$ admits a ``yes'' instance, and the result~follows.  
\end{proof}
\end{theorem}	

\subsection{MILP-Oracle Computable Upper Bounds} \label{subsec: np upper bounds}
In this section, we analyze the upper bound (\ref{eq: upper bound}) and the tightest upper bound (\ref{eq: tightest upper bound}) for $z^*_{IBLP}$. Similar to [\textbf{IL-D}], we introduce the following decision problem (recall Definition \ref{definition 1}):
\begin{itemize}
	\item[$ $] [\textbf{IU-D}]: Given a class $\mathcal{C}$ of IBLPs satisfying Assumption \textbf{A1$'$}, does there exist an MILP-oracle computable bound~$\beta_{\,U}$ such~that
	\begin{equation} \nonumber
		z^*_{\,IBLP}(I) \leq \beta_{\,U}(I) \leq \hat{z}^*_{\,U}(I) \quad \forall I \in \mathcal{C},
	\end{equation}
	with strict inequality $\beta_{\,U}(I) < \hat{z}^*_{\,U}(I)$ whenever $z_{\,IBLP}^*(I) < \hat{z}^*_{\,U}(I)$?
\end{itemize}

Following the discussion in Section \ref{subsec: upper bounds}, we observe that  $\hat{z}^*_{\,U}(I)$ can be computed by solving two polynomial-size MILPs given that both [\textbf{SLR}] with $n_1 = m_1 = 0$ and the associated follower's problem in (\ref{eq: follower's optimal response}) admit unique optimal solutions. The following result holds. 

\begin{theorem} \label{theorem 6}
	Let $\mathcal{C}$ be restricted to the class of min-max IBLPs satisfying Assumption~\textbf{A1$\,'$}, for which both the single-level relaxation~\upshape[\textbf{SLR}] \itshape and the associated follower's problem in \upshape (\ref{eq: follower's optimal response}) \itshape admit unique optimal solutions.
	Then, unless $\Delta_2^P = \Sigma^P_2$, the upper bound $\hat{z}^*_{\,U}$ is provably best in the sense of \upshape [\textbf{IU-D}]. \itshape 
	\begin{proof}
    Assume that the answer to [\textbf{IU-D}] is positive. Then there
    exists an MILP-oracle computable bound $\beta_{\,U}$ such that $z^*_{\,IBLP}\le \beta_{\,U}\le \hat z^*_{\,U}$ with strict inequality whenever $z^*_{\,IBLP}<\hat z^*_{\,U}$.
    We show that this assumption implies $\Sigma_2^P=\Delta_2^P$. To this end we use a reduction from QSAT$_2$ given by:
	\[
	\exists \mathbf{x}\in\{0,1\}^{n_x}\;
	\forall \mathbf{y}\in\{0,1\}^{n_y}\;
	\varphi(\mathbf{x},\mathbf{y}),
	\]
	where $\varphi$ is in 3-disjunctive normal form (3-DNF) with $m$ terms. 
	
	We introduce the following instance of~[\textbf{IBLP}] associated with QSAT$_2$:
	\begin{align}	 \label{QSAT min-max integer 2}
		z^*_{\,IBLP} = \min_{(\mathbf{x}, s) \in \tilde{X}} \; 
		\Biggl\{ s + \max_{\mathbf{y}, v, w}  \Big\{ & 2w + v:  \text{(\ref{cons: QSAT min-max integer 2})--(\ref{cons: QSAT min-max integer 1})}, \nonumber  \\
		& w \in \{0, 1\}, \quad w \leq 1 - s,\\
		& v \leq s, \quad y_i \leq s \quad \forall i \in [n_y] \nonumber \Big\}\Biggr\},  
	\end{align}
where 
\[\tilde{X} := \Big\{(\mathbf{x}, s) \in \{0, 1\}^{n_x + 1}: \;  x_i \leq s \quad \forall i \in [n_x] \Big\}. \] 
Notably, the single-level relaxation [\textbf{SLR}] of (\ref{QSAT min-max integer 2}) admits the unique optimal solution obtained by setting all variables $\mathbf{x}, \mathbf{y}, s, v, w$ equal to zero, with the optimal objective function value $z^*_{SLR} = 0$. The corresponding follower's optimal solution in (\ref{QSAT min-max integer 2}) is unique and is given by $\mathbf{y}^* = \mathbf{0}$, $v^* = 0$ and~$w^* = 1$. This yields the tightest upper bound $\hat{z}^*_{\,U} = 2$; recall (\ref{eq: tightest upper bound}).

First, assume that QSAT$_2$ admits a ``yes'' instance. Then, there exists
$\tilde{\mathbf{x}} \in \{0,1\}^{n_x}$ such that~$\varphi(\tilde{\mathbf{x}}, \mathbf{y})$ is satisfied for every
$\mathbf{y} \in \{0,1\}^{n_y}$. By setting $\mathbf{x} = \tilde{\mathbf{x}}$ and $s = 1$, the respective follower's optimal solution in (\ref{QSAT min-max integer 2}) yields $w^* = 0$ and $v^* = 0$; recall the proof of Theorem \ref{theorem 5}. Hence,~$z^*_{\,IBLP} \leq 1 < \hat{z}^*_{\,U}$.

Assume that QSAT$_2$ admits a ``no'' instance. Then, for every
$\mathbf{x}\in\{0,1\}^{n_x}$, there exists $\tilde{\mathbf{y}} \in \{0,1\}^{n_y}$ such that $\varphi(\mathbf{x}, \tilde{\mathbf{y}})$ is not satisfied. If $s = 0$, then the unique feasible solution of the leader in (\ref{QSAT min-max integer 2}) is given by $\mathbf{x} = \mathbf{0}$, yielding the objective function value $2$. Otherwise, if $s = 1$, then (\ref{QSAT min-max integer 2}) reduces to
	\begin{align} \label{QSAT min-max integer 3}	
	    \tilde{z}^*_{IBLP} := \min_{\mathbf{x} \in \{0, 1\}^{n_x}} \; 
		\left\{1 + \max_{\mathbf{y}, v}  \Big\{v: \text{(\ref{cons: QSAT min-max integer 2})--(\ref{cons: QSAT min-max integer 1})} \Big\}\right\}. \end{align}
Based on the proof of Theorem \ref{theorem 5}, the optimal objective function value of (\ref{QSAT min-max integer 3}) satisfies $\tilde{z}^*_{IBLP} \geq 2$. By combining the cases $s = 0$ and $s = 1$ we conclude that $z^*_{IBLP} = 2 = \hat{z}^*_{\,U}$. 

As a result, $z^*_{IBLP} < \hat{z}^*_{\,U}$ and consequently $\beta_{\,U} < \hat{z}^*_{\,U}$ if and only if QSAT$_2$ admits a ``yes'' instance.  Since both $\beta_{\,U}$ and $\hat{z}^*_{\,U}$ are MILP-oracle computable, QSAT$_2$ can be decided in $\Delta_2^P$, which contradicts the assumption that $\Delta_2^P\neq\Sigma_2^P$. This observation concludes the proof. 
\end{proof}
\end{theorem}

Taken together, Theorems~\ref{theorem 5} and \ref{theorem 6} demonstrate that the standard SLR-based lower and upper bounds for min-max IBLPs cannot, in general, be uniformly improved within the computational framework of polynomial-time MILP-oracle algorithms. Put differently, unless the polynomial hierarchy collapses, no generic MILP-based decomposition or cutting-plane framework can systematically strengthen these bounds within a polynomial number of~iterations.

\section{Conclusion} \label{sec: conclusions} 

In this paper, we study standard lower and upper bounds for mixed-integer bilevel linear programs obtained by relaxing the follower's optimality condition. Informally, we investigate whether these bounds can be uniformly improved over a class of bilevel problems without a substantial increase in the computational effort required to obtain them. For both continuous and pure integer bilevel linear programs, our complexity-theoretic results show that such uniform improvements are generally impossible within the respective computational regimes, even for the restrictive class of min-max problems. At the same time, we establish that the gap between the standard lower and upper bounds decreases at least linearly as the leader's and the follower's objective vectors become aligned.

Overall, our results provide a complexity-theoretic justification for the use of the standard lower and upper bounds in exact algorithms for bilevel optimization. Although these bounds can be weak for individual instances, uniformly stronger bounds cannot, in general, be obtained at comparable computational cost. This, however, does not preclude stronger bounds for particular instances or more restrictive problem classes, and identifying structural conditions that permit such improvements constitutes a natural direction for future research.

 \singlespacing
 
 \bibliographystyle{apa}
 \bibliography{bibliography}

\begin{thebibliography}{}

\bibitem[\protect\astroncite{Arora and Barak}{2009}]{Arora2009}
Arora, S. and Barak, B. (2009).
\newblock {\em Computational Complexity: A Modern Approach}.
\newblock Cambridge University Press.

\bibitem[\protect\astroncite{Audet et~al.}{1997}]{Audet1997}
Audet, C., Hansen, P., Jaumard, B., and Savard, G. (1997).
\newblock Links between linear bilevel and mixed 0--1 programming problems.
\newblock {\em Journal of Optimization Theory and Applications},
  93(2):273--300.

\bibitem[\protect\astroncite{Baringo and Conejo}{2012}]{Baringo2012}
Baringo, L. and Conejo, A.~J. (2012).
\newblock Transmission and wind power investment.
\newblock {\em IEEE Transactions on Power Systems}, 27(2):885--893.

\bibitem[\protect\astroncite{Ben-Ayed et~al.}{1992}]{Ben1992}
Ben-Ayed, O., Blair, C.~E., Boyce, D.~E., and LeBlanc, L.~J. (1992).
\newblock Construction of a real-world bilevel linear programming model of the
  highway network design problem.
\newblock {\em Annals of Operations Research}, 34(1):219--254.

\bibitem[\protect\astroncite{Bodlaender et~al.}{1990}]{Bodlaender1990}
Bodlaender, H.~L., Gritzmann, P., Klee, V., and Van~Leeuwen, J. (1990).
\newblock Computational complexity of norm-maximization.
\newblock {\em Combinatorica}, 10(2):203--225.

\bibitem[\protect\astroncite{Borrero et~al.}{2019}]{Borrero2019}
Borrero, J.~S., Prokopyev, O.~A., and Saur{\'e}, D. (2019).
\newblock Sequential interdiction with incomplete information and learning.
\newblock {\em Operations Research}, 67(1):72--89.

\bibitem[\protect\astroncite{Buchheim}{2023}]{Buchheim2023}
Buchheim, C. (2023).
\newblock Bilevel linear optimization belongs to {NP} and admits
  polynomial-size {KKT}-based reformulations.
\newblock {\em Operations Research Letters}, 51(6):618--622.

\bibitem[\protect\astroncite{Busygin and Pasechnik}{2006}]{Busygin2006}
Busygin, S. and Pasechnik, D.~V. (2006).
\newblock On {NP}-hardness of the clique partition--independence number gap
  recognition and related problems.
\newblock {\em Discrete Mathematics}, 306(5):460--463.

\bibitem[\protect\astroncite{Caprara et~al.}{2013}]{Caprara2013}
Caprara, A., Carvalho, M., Lodi, A., and Woeginger, G.~J. (2013).
\newblock A complexity and approximability study of the bilevel knapsack
  problem.
\newblock In {\em Integer Programming and Combinatorial Optimization: 16th
  International Conference, IPCO 2013, Valpara{\'\i}so, Chile, March 18-20,
  2013. Proceedings 16}, pages 98--109. Springer.

\bibitem[\protect\astroncite{Caprara et~al.}{2014}]{Caprara2014}
Caprara, A., Carvalho, M., Lodi, A., and Woeginger, G.~J. (2014).
\newblock A study on the computational complexity of the bilevel knapsack
  problem.
\newblock {\em SIAM Journal on Optimization}, 24(2):823--838.

\bibitem[\protect\astroncite{Caprara et~al.}{2016}]{Caprara2016}
Caprara, A., Carvalho, M., Lodi, A., and Woeginger, G.~J. (2016).
\newblock Bilevel knapsack with interdiction constraints.
\newblock {\em INFORMS Journal on Computing}, 28(2):319--333.

\bibitem[\protect\astroncite{Colson et~al.}{2007}]{Colson2007}
Colson, B., Marcotte, P., and Savard, G. (2007).
\newblock An overview of bilevel optimization.
\newblock {\em Annals of Operations Research}, 153(1):235--256.

\bibitem[\protect\astroncite{Dempe}{2002}]{Dempe2002}
Dempe, S. (2002).
\newblock {\em Foundations of Bilevel Programming}, volume~61 of {\em Nonconvex
  Optimization and Its Applications}.
\newblock Springer, Dordrecht.

\bibitem[\protect\astroncite{DeNegre}{2011}]{DeNegre2011}
DeNegre, S. (2011).
\newblock {\em Interdiction and Discrete Bilevel Linear Programming}.
\newblock Lehigh University.

\bibitem[\protect\astroncite{Deng}{1998}]{Deng1998}
Deng, X. (1998).
\newblock Complexity issues in bilevel linear programming.
\newblock In Pardalos, P.~M., Dempe, V.~F., and Migdalas, A.~A., editors, {\em
  Multilevel Optimization: Algorithms and Applications}, pages 149--164.
  Springer, Boston, MA.

\bibitem[\protect\astroncite{Fischetti et~al.}{2017}]{Fischetti2017}
Fischetti, M., Ljubi{\'c}, I., Monaci, M., and Sinnl, M. (2017).
\newblock A new general-purpose algorithm for mixed-integer bilevel linear
  programs.
\newblock {\em Operations Research}, 65(6):1615--1637.

\bibitem[\protect\astroncite{Fontaine and Minner}{2014}]{Fontaine2014}
Fontaine, P. and Minner, S. (2014).
\newblock Benders decomposition for discrete--continuous linear bilevel
  problems with application to traffic network design.
\newblock {\em Transportation Research Part B: Methodological}, 70:163--172.

\bibitem[\protect\astroncite{Garey and Johnson}{1979}]{Garey1979}
Garey, M.~R. and Johnson, D.~S. (1979).
\newblock {\em Computers and Intractability: A Guide to the Theory of
  NP-Completeness}.
\newblock W. H. Freeman and Company, New York.

\bibitem[\protect\astroncite{Hansen et~al.}{1992}]{Hansen1992}
Hansen, P., Jaumard, B., and Savard, G. (1992).
\newblock New branch-and-bound rules for linear bilevel programming.
\newblock {\em SIAM Journal on Scientific and Statistical Computing},
  13(5):1194--1217.

\bibitem[\protect\astroncite{Jeroslow}{1985}]{Jeroslow1985}
Jeroslow, R.~G. (1985).
\newblock The polynomial hierarchy and a simple model for competitive analysis.
\newblock {\em Mathematical Programming}, 32(2):146--164.

\bibitem[\protect\astroncite{Kahruman-Anderoglu et~al.}{2016}]{Kahruman2016}
Kahruman-Anderoglu, S., Buchanan, A., Butenko, S., and Prokopyev, O.~A. (2016).
\newblock On provably best construction heuristics for hard combinatorial
  optimization problems.
\newblock {\em Networks}, 67(3):238--245.

\bibitem[\protect\astroncite{Ketkov and Prokopyev}{2026}]{Ketkov2026}
Ketkov, S.~S. and Prokopyev, O.~A. (2026).
\newblock On the complexity of bilevel linear and quadratic programs in fixed
  dimensions.
\newblock {\em arXiv preprint arXiv:2511.15592}.

\bibitem[\protect\astroncite{Kleinert et~al.}{2021}]{Kleinert2021}
Kleinert, T., Labb{\'e}, M., Ljubi{\'c}, I., and Schmidt, M. (2021).
\newblock A survey on mixed-integer programming techniques in bilevel
  optimization.
\newblock {\em EURO Journal on Computational Optimization}, 9:100007.

\bibitem[\protect\astroncite{K{\"o}ppe et~al.}{2010}]{Koppe2010}
K{\"o}ppe, M., Queyranne, M., and Ryan, C.~T. (2010).
\newblock Parametric integer programming algorithm for bilevel mixed integer
  programs.
\newblock {\em Journal of Optimization Theory and Applications},
  146(1):137--150.

\bibitem[\protect\astroncite{Lov{\'a}sz}{1979}]{Lovasz1979}
Lov{\'a}sz, L. (1979).
\newblock On the shannon capacity of a graph.
\newblock {\em IEEE Transactions on Information Theory}, 25(1):1--7.

\bibitem[\protect\astroncite{Moore and Bard}{1990}]{Moore1990}
Moore, J.~T. and Bard, J.~F. (1990).
\newblock The mixed integer linear bilevel programming problem.
\newblock {\em Operations Research}, 38(5):911--921.

\bibitem[\protect\astroncite{Sinha et~al.}{2017}]{Sinha2017}
Sinha, A., Malo, P., and Deb, K. (2017).
\newblock A review on bilevel optimization: From classical to evolutionary
  approaches and applications.
\newblock {\em IEEE Transactions on Evolutionary Computation}, 22(2):276--295.

\bibitem[\protect\astroncite{Stockmeyer}{1976}]{Stockmeyer1976}
Stockmeyer, L.~J. (1976).
\newblock The polynomial-time hierarchy.
\newblock {\em Theoretical Computer Science}, 3(1):1--22.

\bibitem[\protect\astroncite{Wiesemann et~al.}{2013}]{Wiesemann2013}
Wiesemann, W., Tsoukalas, A., Kleniati, P.-M., and Rustem, B. (2013).
\newblock Pessimistic bilevel optimization.
\newblock {\em SIAM Journal on Optimization}, 23(1):353--380.

\bibitem[\protect\astroncite{Wogrin et~al.}{2020}]{Wogrin2020}
Wogrin, S., Pineda, S., and Tejada-Arango, D.~A. (2020).
\newblock Applications of bilevel optimization in energy and electricity
  markets.
\newblock In {\em Bilevel Optimization: Advances and Next Challenges}, pages
  139--168. Springer.

\bibitem[\protect\astroncite{Yue and You}{2017}]{Yue2017}
Yue, D. and You, F. (2017).
\newblock Stackelberg-game-based modeling and optimization for supply chain
  design and operations: A mixed integer bilevel programming framework.
\newblock {\em Computers \& Chemical Engineering}, 102:81--95.

\bibitem[\protect\astroncite{Zare et~al.}{2019}]{Zare2019}
Zare, M.~H., Borrero, J.~S., Zeng, B., and Prokopyev, O.~A. (2019).
\newblock A note on linearized reformulations for a class of bilevel linear
  integer problems.
\newblock {\em Annals of Operations Research}, 272(1):99--117.

\end{thebibliography}
\end{document}